\documentclass[submission,copyright,creativecommons]{eptcs}
\usepackage{graphicx} % Required for inserting images
\usepackage{mathtools,amssymb}
\usepackage{proof}
\usepackage{listings}
\usepackage{tikz-cd}
\usepackage{amsthm}
\usepackage{enumitem}
\providecommand{\event}{ACT 2025}
\newtheorem{theorem}{Theorem}[section]
\newtheorem{prop}[theorem]{Proposition}
\newtheorem{lemma}[theorem]{Lemma}
\newtheorem{corollary}[theorem]{Corollary}
\theoremstyle{definition}
\newtheorem{definition}{Definition}[section]

\let\bb\mathbb
\let\x\times
\let\brhd\blacktriangleright
\let\ox\otimes
\newcommand\utr{u^R_\ox}
\newcommand\utl{u^L_\ox}
\let\lc\ulcorner
\let\rc\urcorner

\makeatletter
\tikzcdset{
  eq node/.style={
    commutative diagrams/math mode=false, anchor=center},
  eq/.style={
    phantom,
    /tikz/every to/.append style={
      edge node={node[commutative diagrams/eq node]
        {\@eqnswtrue\make@display@tag\ltx@label{#1}}}}}}
\makeatother
\title{Actegories, Copowers, and Higher-Order Message Passing Semantics}
\author{Robin Cockett
\institute{University of Calgary\\Calgary, Canada}
\email{robin@ucalgary.ca}
\and
Melika Norouzbeygi
\institute{University of Calgary\\Calgary, Canada}
\email{melika.norouzbeygi@ucalgary.ca}
}

\def\titlerunning{Actegories, Copowers, and Higher-Order Message Passing Semantics}
\def\authorrunning{J.R.B. Cockett, M. Norouzbeygi}

\begin{document}

\maketitle

\begin{abstract}
In this paper we prove that giving a right actegory with hom-objects is equivalent to giving a right-enriched category with copowers. While this result is known in the closed symmetric setting, our contribution extends the equivalence to non-closed and non-symmetric setting. This generalization is motivated by the semantics of higher-order message passing in the \textbf{Categorical Message Passing Language (CaMPL)}, a concurrent language whose semantics is given by a linear actegory. A desirable feature for this language is the support of higher-order processes: processes that are passed as first class citizens between processes. While this ability is already present in any closed linear type systems -- such as {\bf CaMPL}'s -- to support arbitrary recursive process definitions requires the ability to reuse passed processes. Concurrent resources in {\bf CaMPL}, however, cannot be duplicated, thus, passing processes as linear closures does not provide the required flexibility. This means processes must be passed as sequential data and the concurrent side must be {\em enriched} in the sequential side, motivating the technical result of this paper.
\end{abstract}

\section{Introduction}
A right $\bb{A}$-actegory $\bb{X}$ consists of a monoidal category $\bb{A}$ acting on a category $\bb{X}$ via an action functor $\lhd: \bb{X} \times \bb{A} \to \bb{X}$ in a manner compatible with the monoidal structure of $\bb{A}$. If this action functor admits a right adjoint, written as $\forall X\in \bb{X}, X \lhd - \dashv \mathrm{Hom}(X, -): \bb{A} \to \bb{X}$, then $\bb{X}$ is called a right $\bb{A}$-actegory with \emph{hom-objects}.

A category $\bb{X}$ is \emph{right enriched} in the monoidal category $\bb{A}$ if, for every pair of objects $X, Y \in \bb{X}$, there exists a hom-object $\mathrm{Hom}(X, Y) \in \bb{A}$ along with a composition morphism $m: \mathrm{Hom}(X, Y) \otimes \mathrm{Hom}(Y, Z) \to \mathrm{Hom}(X, Z)$ and a unit morphism $\mathrm{id}: I \to \mathrm{Hom}(X, X)$ satisfying the associativity and unitality conditions described in Definition~\ref{enriched-cat-def} below.

A right $\bb{A}$-enriched category $\bb{X}$ is said to be \emph{copowered} if, for each $A \in \bb{A}$ and $X \in \bb{X}$, there exists a morphism $\eta: A \to \mathrm{Hom}(X, X \lhd A)$ in $\bb{A}$, called the \emph{unit of the copower}, such that for every morphism $f: B \to \mathrm{Hom}(X \lhd A, Y)$, there exists a bijective correspondence~\cite{lucyshynv}:
\[
\infer=[]{A \otimes B \xrightarrow[]{(\eta \otimes f) \, m} \mathrm{Hom}(X, Y)}{B \xrightarrow[]{f} \mathrm{Hom}(X \lhd A, Y)}
\]
where $m$ denotes composition in $\bb{A}$. In the case that $\bb{A}$ is closed, this is equivalent to the existence of an $\bb{A}$-natural isomorphism $\gamma: \mathrm{Hom}(X \lhd A, Y) \cong A \multimap \mathrm{Hom}(X, Y)$~\cite{aNoteActionsOfAMonoidalCategory}.

This paper proves that giving a right $\bb{A}$-actegory with hom-objects is equivalent to giving a right $\bb{A}$-enriched category with copowers. While this result was known in the closed symmetric setting~\cite{aNoteActionsOfAMonoidalCategory} (and more generally in the bicategorical setting \cite{GORDON1997167}), our contribution generalizes it to the non-closed and non-symmetric setting. This extension can be applied to model the semantics of message passing in languages like \textbf{CaMPL}, where being closed is not assumed. Furthermore, we prove that a category is both powered and copowered if and only if the copower functor admits a left parameterized left adjoint, denoted $- \lhd A \dashv A \brhd -$.  Having such a parameterized adjoints is a fundamental feature of the semantics of the message passing \cite{logicOfMessagePassing}.

The motivation for this work arises from the design of the {\bf Ca}tegorical {\bf M}essage {\bf P}assing {\bf L}anguage ({\bf CaMPL}), a concurrent programming language with categorical semantics given by a linear actegory. The sequential side of {\bf CaMPL} is functional in style, while the concurrent side supports message passing between processes along typed channels or protocols. {\bf CaMPL} is based on the work of Cockett and Pastro \cite{logicOfMessagePassing}, which provides both a formal proof theory of message passing and its categorical semantics. It was subsequently implemented by Prashant Kumar \cite{ImplementationOfMPL} and Jared Pon \cite{jaredpon} at the University of Calgary.

The semantics of {\bf CaMPL} lies in a linear actegory \cite{actegoriesAmthmatician}, which is a linearly distributive category \cite{COCKETT1997133} with a distributive monoidal category acting on it in two ways $\lhd: \bb{X}\x\bb{A}\to\bb{X}$ and $\brhd: \bb{A}^{op}\x\bb{X}\to\bb{X}$, such that $\lhd$ is the parameterized left adjoint of $\brhd$  (i.e., for all $A\in\bb{A}$, $-\rhd A \dashv A\brhd-$).\footnote{In \cite{logicOfMessagePassing} the actions $\lhd$ and $\brhd$ are written using $\circ$ and $\bullet$ with $A\circ X := X\lhd A$ and $A\bullet X := A\brhd X$.} These two functors allow the sending and receiving of sequential messages (objects of $\bb{A}$) along concurrent channels (objects of $\bb{X}$). The parameterized adjunction arises from the processes of transmitting or receiving a message, and the fact that this can be done equivalently on input or output polarity channels.

Passing concurrent processes to other processes is already possible in the concurrent side of {\bf CaMPL} as it is a $*$-autonomous category and $*$-autonomous categories are closed:
\[ \infer[]{\Gamma\vdash A^\perp\oplus B := A \multimap B}{\Gamma, A\vdash B} \]
This allows a process $A \multimap B$ to be passed to another process through an input polarity channel of type $A^\perp \oplus B$ and ``applied'' to an input channel $A$ to yield an output channel $B$:
\begin{lstlisting}
    proc apply :: | A, Neg(A) (+) B => B = 
        | a, neg_a_b => b -> 
            fork neg_a_b as
                    neg_a -> neg_a |=| neg(a)
                    b' -> b |=| b'     
\end{lstlisting}
\noindent Here the channel \texttt{neg\_a\_b} of type \texttt{Neg(A)(+)B} (i.e. $A^{\perp}\oplus B$) is forked into two independent processes which, using \texttt{|=|}, respectively identify channels \texttt{neg\_a} with the negation of \texttt{a} and \texttt{b'} with \texttt{b}.

However, concurrent resources cannot be duplicated, thus passing a process on a concurrent channel will not support arbitrary recursive process definitions which may require duplicating the code.  For example, suppose we want to define a process that has input (polarity) channels of type $A$ and $A^\perp \oplus A$ (representing the input process of type $A \multimap A$), and an output (polarity) channel of type $A$  along with a (sequential) counter $n: Int$ whose aim is to apply the input process, $n$ times to the input channel. This suggests the following {\bf CaMPL} code for defining this process:
\begin{lstlisting}
    proc q :: Int | A, Neg(A) (+) A => A =
        n | x, p => y -> case n of
                0 -> x |=| y                      -- base case 
                _ -> plug                         -- recursive case
                        apply(x, p => z)          -- a process that applies p
                        q( n - 1 | z, p => y)     -- recurse
\end{lstlisting}
However, this code snippet is not a valid {\bf CaMPL} program, since the channel $p$ is used twice, violating the linearity constraint on concurrent resources. 

The solution to the problem is to store the concurrent process as sequential data which can then be duplicated and used multiple times. To achieve this, we have introduced two new language constructs:
\begin{itemize}
    \item {\tt store}: Converts a concurrent process into sequential data.
    \item {\tt use}: Allows the calling of a stored process.
\end{itemize}
We define a new sequential data type ${\sf Store}(\Phi|\Gamma\Vdash\Delta)$ to represent stored processes. Using these new constructs, the previous program can be rewritten (correctly) as follows:
\begin{lstlisting}
    proc q :: Int, Store( | A => A) | A => A =
        n, p | x => y -> case n of
                0 -> x |=| y                      -- base case 
                _ -> plug                         -- recursive case
                        use(p)( | x => z)         -- using the stored process
                        q( n - 1, p | z => y)     -- recurse
\end{lstlisting}

\noindent The categorical machinery which enables the storing of concurrent processes as sequential data is given by an enrichment of the concurrent world in the sequential world.  This in turn is abstractly given by the equivalence between an actegory with hom-objects and an enriched category with copower, as is proved in this paper.
\\\\
\noindent{\bf Acknowledgment:} The authors would like to thank Rory Lucyshyn-Wright for pointing out that our main technical result does not require a closed setting.
\section{On Actegories, Enrichment, and Copowers}
%\subsection{Actegories}
% \begin{definition}
%     Let \(\bb{X}, \bb{Y}, \bb{Z}\) be categories. An \textbf{\(\bb{X}\)-parameterized adjunction}\cite{mac1998categories} between \(\bb{Y}\) and \(\bb{Z}\) consists of a pair of functors $F: \bb{Z} \times \bb{X} \to \bb{Y}$ and $G : \bb{X}^{\mathrm{op}} \times \bb{Y} \to \bb{Z}$ such that for every object \(X \in \bb{X}\), the functor
% $
% F_X := F(-, X) : \bb{Z}\to\bb{Y}
% $
% has a right adjoint
% $
% G_X := G(X, -) : \bb{Y}\to\bb{Z},
% $
% and the adjunction bijection
% $$
% \infer=[]{
% Z\to G(X, Y)
% }{F(Z, X)\to Y}
% $$
% \noindent is natural in \(Y \in \bb{Y}\), and \(Z \in \bb{Z}\).
% \end{definition}

\begin{definition}\label{actegory-def}
A \textbf{right $\bb{A}$-actegory} consists of a monoidal category $\bb{A}$, a category $\bb{X}$, an action functor $\lhd : \bb{X}\x\bb{A}\to\bb{X}$ and two natural isomorphisms $a_\lhd: X \lhd (A \ox B) \to (X \lhd A) \lhd B$ and $u_\lhd: X\lhd I \to X$ satisfying the following coherences \cite{actegoriesAmthmatician}:
\begin{center}
\footnotesize
    \begin{tikzcd}
        X\lhd (A \ox I) \arrow[r, "a_\lhd"] \arrow[d, " 1 \lhd u^R_\ox"']
        & (X \lhd A) \lhd I \arrow[ld, "u_\lhd"] \\
        X \lhd A 
    \end{tikzcd}
    \begin{tikzcd}
       X\lhd (A\otimes (B\otimes C))\arrow[r, "1 \lhd a_\ox"] \arrow[d, "a_\lhd"']
            & X \lhd ((A\otimes B) \otimes C)\arrow[d, "a_\lhd"] \\
            (X \lhd A) \lhd (B\otimes C )\arrow[d, "a_\lhd "']
            & (X \lhd (A\otimes B))\lhd C \arrow[ld, "a_\lhd\lhd 1"]\\
            ((X\lhd A)\lhd B)\lhd C
    \end{tikzcd}
    \begin{tikzcd}
        X \lhd (I \ox A)\arrow[r, "a_\lhd"] \arrow[d, "1 \lhd u^L_\otimes"']
        & (X \lhd I) \lhd A \arrow[ld, "u_\lhd \lhd 1"] \\
        X \lhd A
    \end{tikzcd}
\end{center}
\end{definition}
An example of an actegory, is \underline{Set}, with an action on itself given by the cartesian product. More generally any monoidal category acts on itself by the tensor product.

Similar to right $\bb{A}$-actegories, left $\bb{A}$-actegories are equipped with an action functor $\rhd : \bb{A}\x\bb{X}\to\bb{X}$ along with two natural isomorphisms $a_\rhd : (A\ox B)\rhd X \to A\rhd (B\rhd X)$ and $u_\rhd: I\rhd X \to X$ satisfying similar coherences.

A right $\bb{A}$-actegory is precisely a left $\bb{A}^{rev}$-actegory, where $\bb{A}^{rev}$ is the monoidal category obtained by equipping the underlying category of $\bb{A}$ with the tensor product $A \ox^{rev} B := B \ox A$ and appropriately adjusting the rest of the structure.

%\subsection{Enrichment}
\begin{definition}
\label{enriched-cat-def}
    A category $\bb{X}$ is a \textbf{right $\bb{A}$-enriched} category  \cite{borceux1994handbook} where $\bb{A}$ is a monoidal category if:
\begin{enumerate}
    \item For every pair of objects $X, Y \in \bb{X}$, there is a hom-object $ \textrm{Hom}(X, Y)\in\mathbb{A}$
    \item For every triple $X, Y, Z \in\mathbb{X}$, a ``composition" morphism \[m_{XYZ} :  \textrm{Hom}(X, Y) \otimes  \textrm{Hom}(Y, Z) \rightarrow  \textrm{Hom}(X, Z) \] exists in $\mathbb{A}$.
    \item for every object $X \in \mathbb{X}$ there is a ``unit" morphism $id: I \rightarrow  \textrm{Hom}(X, X)$ in $\mathbb{A}$.
Such that the following associativity and unit diagrams commute:
%\begin{enumerate}
    %\item Associativity of the composition:
    \begin{center}
    \scriptsize
    \hspace*{-1cm}
        \begin{tikzcd}[column sep = huge]
            &(\textrm{Hom}(X, Y)\otimes \textrm{Hom}(Y, Z))\otimes \textrm{Hom}(Z, W) 
            \arrow[rdd, "m_{XYZ}\otimes 1_{ \textrm{Hom}(Z, W)}"] \arrow[ldd, "a_\ox"']\\\\
            \textrm{Hom}(X, Y)\otimes (\textrm{Hom}(Y, Z))\otimes \textrm{Hom}(Z, W)) 
            \arrow[dd, "1_{\textrm{Hom}(X, Y)}\otimes m_{YZW}"']  
            &&\bb{X}(X, Z)\otimes\textrm{Hom}(Z, W)\arrow[dd, "m_{XZW}"] \\\\
            \textrm{Hom}(X, Y) \otimes\textrm{Hom}(Y, W) \arrow[rr, "m_{XYW}"'] 
            &&\textrm{Hom}(X, W)
        \end{tikzcd}
    \end{center}
    %\item Identity on the left and right:
    \begin{center}
    \footnotesize
        \begin{tikzcd}[column sep = huge]
            I \otimes  \textrm{Hom}(X, Y) \arrow[rd, "u_{\otimes}^L"'] \arrow[r, "id_x \otimes 1_{ \textrm{Hom}(X, Y)}"]
                &  \textrm{Hom}(X, X) \otimes  \textrm{Hom}(X, Y)\arrow[d, "m_{XXY}"] \\ & \textrm{Hom}(X, Y)
        \end{tikzcd}~~~
        \begin{tikzcd}[column sep = huge]
             \textrm{Hom}(X, Y) \otimes I \arrow[rd, "u_{\otimes}^R"'] \arrow[r, "1_{ \textrm{Hom}(X, Y)} \otimes id_x"]
                &  \textrm{Hom}(X, Y)\otimes  \textrm{Hom}(X, X) \arrow[d, "m_{XXY}"] \\ & \textrm{Hom}(X, Y)
        \end{tikzcd}
    \end{center}
\end{enumerate}
%\end{enumerate}
\end{definition}
In contrast to right enrichment, left enrichment is defined using the same data and conditions, with the only difference being that the "composition" morphism is given by $m_{XYZ} : \text{Hom}(Y, Z) \otimes \text{Hom}(X, Y) \to \text{Hom}(X, Z)$.
\begin{definition}
The \textbf{underlying ordinary category} of the $\bb{A}$-enriched category $\bb{X}$ is denoted by $\lc\bb{X}\rc$ and has the same objects as $\bb{X}$ while a map $f: X\to Y$ in $\lc\bb{X}\rc$ is just an element $\lc f\rc: I\to  \textrm{Hom}(X, Y)$. The composition of two maps $f:X\to Y$ and $g: Y\to Z$ in $\lc\bb{X}\rc$ is given by \cite{kelly1982basic}:
$$I \xrightarrow[]{(\utr)^{-1}} I\ox I \xrightarrow[]{\lc f\rc\ox \lc g\rc}  \textrm{Hom}(X, Y)\ox \textrm{Hom}(Y, Z)\xrightarrow[]{m}  \textrm{Hom}(X, Z)$$
and the identity $1_A: X\to X$ is $id: I\to  \textrm{Hom}(X, X)$.
\end{definition}
\begin{lemma}\label{swap_m} In any $\bb{A}$-enriched category $\bb{X}$, for all $f: X\to Y$ we have $(1\ox {\rm Hom}(1, f))m = m {\rm Hom}(1, f)$. 
\end{lemma}
\vspace{0.2cm}
% \subsection{Parametrized Adjunction}
% \subsection{Closed Monoidal Categories}
% \begin{definition}
% A monoidal category $\bb{A}$ is said to be \textbf{closed} \cite{barr1990category} if for each
% object $A$ of $\bb{A}$, the functor $A\ox -$ has a right adjoint. If we denote the value at $C$ of the a adjoint by $A\multimap C$ ($A$ appears as a parameter since this is defined for each object $A$), then the condition for $\bb{A}$ to be \textbf{right closed} is the existance of the following bijective correspondence, natural in $B$ and $C$,

% $$
% \infer=[]{B\xrightarrow[]{g = \textrm{curry}(f)} A\multimap C}{A\ox B\xrightarrow[]{f = \textrm{uncurry}(g)} C}
% $$
% \end{definition}
%\subsection{Copowers}
Classically copowers are defined for right-closed categories, as described in the introduction.  However in this paper we will use the definition of copowers for non-closed categories:
\begin{definition}\label{copowers-def}
    An $\bb{A}$-enriched category $\bb{X}$ has \textbf{copowers} \cite{pare2012mealy} if for any $A\in\bb{A}$ and $X\in\bb{X}$, there exists an object $X\lhd A \in \bb{X}$ and a morphism $\eta: A\to  \textrm{Hom}(X, X\lhd A)$ in $\bb{A}$ (called the unit of the copower) such that for any $B\in\bb{A}$ and $Y\in\bb{X}$ there is a bijective correspondence:
    $$
    \infer=[]{A\otimes B \xrightarrow[]{g = (f)^\Downarrow = (\eta\ox f)m}  \textrm{Hom}(X, Y)}{B \xrightarrow[]{f = (g)^\Uparrow}\textrm{Hom}(X\lhd A, Y)}
    $$
\end{definition}
Another way to define copowers is by using two combinators $(-)^\Downarrow$ and $(-)^\Uparrow$
which satisfy the following properties:
\begin{center}
\begin{tabular}{ll}
(i)~$(1\ox h)(f)^\Downarrow = (hf)^\Downarrow$ & (ii) $h(f)^\Uparrow = ((1\ox h)f)^\Uparrow$ \\
(iii)~$((f)^\Downarrow \ox g)m = a_\ox((f\ox g)m)^\Downarrow$ &
(iv) $((f)^\Uparrow \ox g)m = (a_\ox(f\ox g)m)^\Uparrow$
\end{tabular}
\end{center}
%\begin{enumerate}[label=(\roman{*})]
%    \item $(1\ox h)(f)^\Downarrow = (hf)^\Downarrow$
%    \item $h(f)^\Uparrow = ((1\ox h)f)^\Uparrow$
%    \item $((f)^\Downarrow \ox g)m = a_\ox((f\ox g)m)^\Downarrow$ 
%    \item $((f)^\Uparrow \ox g)m = (a_\ox(f\ox g)m)^\Uparrow$
%\end{enumerate}
This definition is used throughout the proofs in this paper and is equivalent to the definition of copowers given in \ref{copowers-def}, as shown in the following lemma.
% \ref{eta_down_up_eq}. 
\begin{lemma}\label{eta_down_up_eq}
Given a bijective correspondence 
$$
\infer=[]{A\otimes B \xrightarrow[]{g = (f)^\Downarrow}  \textrm{Hom}(X, Y)}{B \xrightarrow[]{f = (g)^\Uparrow}\textrm{Hom}(X\lhd A, Y)}
$$
with the two combinators $(-)^\Downarrow$ and $(-)^\Uparrow$, the following are equivalent:
\begin{enumerate}[label=(\roman{*})]
    \item $\exists~\eta: A\to \textrm{Hom}(X, X\lhd A)$ and $(f)^\Downarrow = (\eta\ox f)m$\label{fdown-with-eta}
    \item The following properties hold for $(-)^\Downarrow$ and $(-)^\Uparrow$:\label{down-up-properties}
    \begin{enumerate}
        \item $(1\ox h)(f)^\Downarrow = (hf)^\Downarrow$
        \item $h(f)^\Uparrow = ((1\ox h)f)^\Uparrow$
        \item $((f)^\Downarrow \ox g)m = a_\ox((f\ox g)m)^\Downarrow$ 
        \item $((f)^\Uparrow \ox g)m = (a_\ox(f\ox g)m)^\Uparrow$
    \end{enumerate}
\end{enumerate}

\end{lemma}
\begin{proof}
We start with proving $\ref{fdown-with-eta}\Rightarrow\ref{down-up-properties}$. We can define $\eta := (\utr)^{-1}(id)^\Downarrow$ using the given bijection as below:
$$
\infer=[]{
\infer=[]{A\xrightarrow[]{(\utr)^{-1}}A\ox I\xrightarrow[]{(id)^\Downarrow} \textrm{Hom}(X, X\lhd A)}
{A\ox I \xrightarrow[]{(id)^\Downarrow} \textrm{Hom}(X, X\lhd A)}
}
{I\xrightarrow[]{id} \textrm{Hom}(X\lhd A, X\lhd A)}
$$
We have:
\begin{align*}
    (\eta\ox f)m 
    &= (((\utr)^{-1}(id)^\Downarrow)\ox f)m\\
    &= ((\utr)^{-1}\ox 1)(1\ox f)((id)^\Downarrow\ox 1)m\\
    &= ((\utr)^{-1}\ox 1)(1\ox f)a_\ox((id\ox 1)m)^\Downarrow\\
    &= ((\utr)^{-1}\ox 1)a_\ox(1\ox f)((id\ox 1)m)^\Downarrow\\
    &= ((\utr)^{-1}\ox 1)a_\ox(1\ox f)((\utl)^{-1})^\Downarrow\\
    &= ((\utr)^{-1}\ox 1)a_\ox(f(\utl)^{-1})^\Downarrow\\
    &= (1\ox (\utl)^{-1})(f(\utl)^{-1})^\Downarrow\\
    &= ((\utl)^{-1}f(\utl)^{-1})^\Downarrow\\
    &= (f(\utl)^{-1}\utl)^\Downarrow\\
    &= (f)^\Downarrow
\end{align*}
\\
Next, we prove $\ref{down-up-properties}\Rightarrow\ref{fdown-with-eta}$. We have:
    \begin{align*}
    (a)~~~~&\vspace{-0.75cm}
        (1\ox h)(f)^\Downarrow
        = (1\ox h)(\eta\ox f)m
        = (\eta\ox hf)m
        = (hf)^\Downarrow\\
    % \end{align*}
    % \begin{align*}\\
    (b)~~~~&h(g)^\Uparrow= ((h(g)^\Uparrow)^\Downarrow)^\Uparrow = ((\eta\ox h(g)^\Uparrow)m)^\Uparrow\\
        &= ((1\ox h)(\eta\ox(g)^\Uparrow)m)^\Uparrow\\
        &= ((1\ox h)((g)^\Uparrow)^\Downarrow)^\Uparrow\\
        &= ((1\ox h)g)^\Uparrow\\
    %\end{align*}
    %\begin{align*}
     (c)~~~~&((f)^\Downarrow\ox g)m
     = ((\eta\ox f)m\ox g)m\\
     &= a_\ox (\eta\ox (f\ox g)m)m\\
     &= a_\ox ((f\ox g)m)^\Downarrow\\
    % \end{align*}
    % \begin{align*}
     (d)~~~~&((f)^\Uparrow\ox g)m
     = ((((f)^\Uparrow\ox g)m)^\Downarrow)^\Uparrow\\
     &= ((\eta\ox ((f)^\Uparrow\ox g)m)m)^\Uparrow\\
     &= (a_\ox ((\eta\ox (f)^\Uparrow)m \ox g)m)^\Uparrow\\
     &= (a_\ox (((f)^\Uparrow)^\Downarrow \ox g)m)^\Uparrow\\
     &= (a_\ox (f\ox g)m)^\Uparrow
    \end{align*}
\end{proof}

\begin{corollary}\label{updown-all}
    Given the bijective correspondence discussed in Lemma $\ref{eta_down_up_eq}$, the following properties hold for $(-)^\Uparrow$ and $(-)^\Downarrow$:
    \begin{enumerate}[label=(\roman{*})]
        \item $(1\ox h)(f)^\Downarrow \textrm{Hom}(1, k) = (hf \textrm{Hom}(1, k))^\Downarrow$
        \item $h(g)^\Uparrow \textrm{Hom}(1, k) = ((1\ox h)g \textrm{Hom}(1, k))^\Uparrow$
    \end{enumerate}
\end{corollary}

%%%%%%%%%%%%%%%%%%%%%%%%%%%%%%%%%%%%%%%%%%%%%%%%%%%%%%%%%%%%%%%%%%%%%%%%%%
\section{Right actegories with hom-objects have copowers}
\begin{prop}\label{actegories-w-hom-objs-give-copowers}
    If $\bb{X}$ is a right $\bb{A}$-actegory with hom-objects, that is for any $X\in\bb{X}$ there is an adjunction $X\lhd - \dashv \textrm{Hom}(X, -): \bb{A}\to\bb{X}$ then $\mathbb{X}$ is a right $\bb{A}$-enriched category with copowers.
\end{prop}
\begin{proof}
We start with proving the enrichment.
In the right $\bb{A}$-actegory $\bb{X}$, we take `$m$' to be the composition morphisms and `$id$' to be the unit as below:\\
$$
\footnotesize
\begin{tikzcd}
    X\lhd \textrm{Hom}(X, Z)\arrow[r, "\epsilon"] 
        & Z \\
    X \lhd \textrm{Hom}(X,Y) \otimes \textrm{Hom}(Y,Z)\arrow[u, dashed, "1 \lhd m"] \arrow[ur, "a_\lhd (\epsilon \lhd 1) \epsilon"']
\end{tikzcd}
\hspace{1cm}
\begin{tikzcd}
    X\lhd \textrm{Hom}(X, X)\arrow[r, "\epsilon"] 
        & X \\
    X \lhd I \arrow[u, dashed, "1 \lhd id"] \arrow[ur, "u_\lhd"']
\end{tikzcd}
$$
In order to have an enrichment, we have to show the associativity and identity rules:
\begin{enumerate}
    \item Identity on the left:
    $$
    \scriptsize
    \begin{tikzcd}
    \arrow[drr, eq=eq:d1] 
    X \lhd \textrm{Hom}(X, Y)\arrow[rr, "\epsilon"] 
    && Y\\
    X \lhd \textrm{Hom}(X, X) \ox \textrm{Hom}(X, Y)\arrow[u, "1\lhd m"] \arrow[r, "a_\lhd"{name=A}] 
    & X \lhd \textrm{Hom}(X, X) \lhd \textrm{Hom}(X, Y)\arrow[r, "\epsilon\lhd 1"]
    &X\lhd \textrm{Hom}(X,Y)\arrow[u, "\epsilon"']\\
    X \lhd I\otimes \textrm{Hom}(X, Y)\arrow[u, "1 \lhd id \otimes 1"] \arrow[r, "a_\lhd"{name=B}]\arrow[urr, "1\lhd u^L_\otimes"'{name=E}, bend right=30]\arrow[, from=A, to=B, eq=eq:d2]
    &X \lhd I \lhd \textrm{Hom}(X, Y) \arrow[u, "1\lhd id \lhd 1"{name=C}] \arrow[ur, " u_\lhd \lhd 1"'{name=D}]\arrow[, from=C, to=D, eq=eq:d3]\\\\
    &\arrow[, from=B, to=E, eq=eq:d4] 
    \end{tikzcd}
    $$
    square \ref{eq:d1} is the definition of $m$, square \ref{eq:d2} commutes because $a_\lhd$ is natural, triangle \ref{eq:d3} is the definition of $id$ and finally \ref{eq:d4} commutes because of \ref{actegory-def}. So the whole diagram above commutes, and thus we have $\tiny (1\lhd id_X \otimes 1)(1\lhd m_{XXY})\epsilon = (1\lhd u^L_\otimes ) \epsilon$.
    Since $\epsilon$ is the counit of the adjunction, the couniversal property gives:
    $$
    \footnotesize
    \begin{tikzcd}
    X \lhd \textrm{Hom}(X, Y)\arrow[r, "\epsilon"] 
        & Y \\
    X \lhd I \otimes \textrm{Hom}(X, Y)\arrow[u, dashed, "1\lhd u^L_\otimes "] \arrow[ur, "(1\lhd u^L_\otimes)\epsilon"']
    \end{tikzcd}
    \hspace{1cm}
    \begin{tikzcd}
    X \lhd \textrm{Hom}(X, Y)\arrow[r, "\epsilon"] 
        & Y \\
    X \lhd I \otimes \textrm{Hom}(X, Y)\arrow[u, dashed, "(1\lhd id \otimes 1)( 1\lhd m)"] \arrow[ur, "(1\lhd id \otimes 1)( 1\lhd m)\epsilon = (1\lhd u^L_\otimes)\epsilon"']
    \end{tikzcd}
    $$
    Since the map from $X \lhd I\otimes \textrm{Hom}(X, Y)$ to $X\lhd \textrm{Hom}(X, Y)$ must be unique, then we have $1 \lhd u^L_\otimes = (1 \lhd id \otimes 1)(1\lhd m_{XXY})$ which proves the left identity rule.
    \item Identity on the right:
    $$
    \scriptsize
    \begin{tikzcd}
    \arrow[drr, eq=eq:d21]
    X \lhd \textrm{Hom}(X, Y)\arrow[rr, "\epsilon"{name=A}] 
    && Y \\
    \arrow[dr, eq=eq:d22]
    X \lhd \textrm{Hom}(X, Y) \otimes \textrm{Hom}(Y, Y)\arrow[u, "1 \lhd m"] \arrow[r, "a_\lhd"]
    & X \lhd \textrm{Hom}(X, Y)\lhd \textrm{Hom}(Y, Y) \arrow[r, "\epsilon\lhd 1"]
    &Y \lhd \textrm{Hom}(Y,Y)\arrow[u, "\epsilon"']\\
    X \lhd \textrm{Hom}(X, Y)\otimes I \arrow[u, "1\lhd 1\otimes id"] \arrow[r, "a_\lhd"{name=C}]\arrow[urr, "1\lhd u^R_\otimes"'{name=D}, bend right=30]
    &X \lhd \textrm{Hom}(X, Y)\lhd I\arrow[u, "1\lhd 1\lhd id_Y"{name=A}] \arrow[ur, "1\lhd u_\lhd"'{name=B}]\arrow[, from=A, to=B, eq=eq:d23]\\\\
    &\arrow[, from=C, to=D, eq=eq:d24]
    \end{tikzcd}
    $$
    square \ref{eq:d21} is the definition of $m$, square \ref{eq:d22} commutes because $a_\lhd$ is natural, triangle \ref{eq:d23} is the definition of $id$ and finally \ref{eq:d24} commutes because of \ref{actegory-def}. Thus the whole diagram above commutes, so we have $(1\lhd 1\otimes id)(1 \lhd m)\epsilon = (1\lhd u^R_\otimes) \epsilon$. Since $\epsilon$ is the counit of the adjunction, the couniversal property gives:
    $$
    \footnotesize
    \begin{tikzcd}
    X \lhd \textrm{Hom}(X, Y) \arrow[r, "\epsilon"] 
        & Y \\
    X \lhd \textrm{Hom}(X, Y)\otimes I \arrow[u, dashed, "1\lhd u^R_\otimes"] \arrow[ur, "(1 \lhd u^R_\otimes)\epsilon"']
    \end{tikzcd}
    \hspace{1cm}
    \begin{tikzcd}
    X \lhd \textrm{Hom}(X, Y)\arrow[r, "\epsilon"] 
        & Y \\
    X\lhd \textrm{Hom}(X, Y)\otimes I \arrow[u, dashed, "(1\lhd 1\otimes id)(1 \lhd m_{XYY})"] \arrow[ur, "(1\lhd 1\otimes id)(1 \lhd m)\epsilon = (1 \lhd u^R_\otimes)\epsilon"']
    \end{tikzcd}
    $$
    Since the map from $X\lhd \textrm{Hom}(X, Y)\otimes I$ to $X \lhd \textrm{Hom}(X, Y) $ must be unique, then we have:
    $1\lhd u^R_\otimes = (1\lhd 1\otimes id_Y)(1 \lhd m_{XYY})$
    which proves the identity on the right.

    \item Associativity of composition is given by:
    \begin{align*}
        (1\lhd (1\otimes m))(1\lhd m)\epsilon
        & = (1\lhd (1\otimes m))a_\lhd(\epsilon\lhd1)\epsilon\tag{def of $m$} \\
        & = a_\lhd(\epsilon\lhd m)\epsilon\tag{naturality of $a_\lhd$}\\
        & = a_\lhd(\epsilon\lhd 1)a_\lhd(\epsilon\lhd 1)\epsilon\tag{def of $m$}\\
        & = a_\lhd a_\lhd((\epsilon\lhd1)\lhd 1)(\epsilon\lhd 1)\epsilon\tag{naturality of $a_\lhd$}\\
        & = (1\lhd a_\otimes)a_\lhd(a_\lhd \lhd 1)((\epsilon\lhd1)\lhd 1)(\epsilon\lhd 1)\epsilon\tag{associativity of $a_\lhd$}\\
        & = (1\lhd a_\otimes)a_\lhd((1\lhd m)\lhd 1)(\epsilon\lhd 1)\epsilon\tag{def of $m$}\\
        & = (1\lhd a_\otimes)(1\lhd (m\ox 1))a_\lhd(\epsilon\lhd 1)\epsilon\tag{naturality of $a_\lhd$}\\
        & = (1\lhd a_\otimes)(1\lhd (m\ox 1))(1\lhd m)\epsilon\tag{def of $m$}\\
    \end{align*}
    Since $\epsilon$ is the counit of the adjunction, the couniversal property gives:
    $$
    \scriptsize
    \begin{tikzcd}
    X\lhd \textrm{Hom}(X, W)\arrow[r, "\epsilon"] 
    & W \\
    X \lhd \textrm{Hom}(X,Y)\otimes \textrm{Hom}(Y,Z)\otimes \textrm{Hom}(Z,W)\arrow[u, dashed, "1\lhd (1\ox m))(1\lhd m)"] \arrow[ur, "(1\lhd (1\ox m))(1\lhd m)\epsilon"']
    \end{tikzcd}
    $$
    $$
    \scriptsize
    \begin{tikzcd}
     X \lhd \textrm{Hom}(X, W)\arrow[r, "\epsilon"] 
     & W \\
     X \lhd (\textrm{Hom}(X,Y)\otimes \textrm{Hom}(Y,Z)\otimes \textrm{Hom}(Z,W)\arrow[u, dashed, "(1\lhd a_\otimes)(1\lhd(m\ox 1))(1\lhd m)"] \arrow[ur, "(1\lhd a_\otimes)(1\lhd(m\ox 1))(1\lhd m)\epsilon = (1\lhd (1\ox m))(1\lhd m)\epsilon"']
    \end{tikzcd}
    $$
    \\
    Since the map from $X \lhd (Hom(X,Y)\otimes \textrm{Hom}(Y,Z)\otimes \textrm{Hom}(Z,W)$ to $X \lhd \textrm{Hom}(X, W)$ must be unique, we have $ (1\lhd a_\otimes)(1\lhd(m_{XYZ}\ox 1))(1\lhd m_{XZW}) = (1\lhd (1\ox m_{YZW}))(1\lhd m_{XYW})$
    which proves the associativity of the composition.
\end{enumerate}
Thus, we have an enrichment of the category $\mathbb{X}$ in the category $\mathbb{A}$.\\
Next, we prove that $\bb{X}$ has copowers over $\bb{A}$. Using the adjunction, we define the $(-)^\Uparrow$ and $(-)^\Downarrow$:
$$
\infer[]{
A\ox B\xrightarrow[]{(a_\lhd f^\sharp)^\flat} \textrm{Hom}(X, Y)
}{X\lhd A\ox B\xrightarrow[]{a_\lhd} \infer[]{X\lhd A\lhd B \xrightarrow[]{f^\sharp} Y}{B\xrightarrow[]{f} \textrm{Hom}(X\lhd A, Y)}}
\hspace{1cm}
\infer[]{
B\xrightarrow[]{(a^{-1}_\lhd g^\sharp)^\flat } \textrm{Hom}(X\lhd A, Y)
}{X\lhd A\lhd B \xrightarrow[]{a^{-1}_\lhd}\infer[]{X\lhd A\ox B\xrightarrow[]{g^\sharp} Y}{A\ox B\xrightarrow[]{g} \textrm{Hom}(X, Y)}}
$$
\\
$(f)^\Downarrow = (a_\lhd f^\sharp)^\flat$ and $(g)^\Uparrow = (a^{-1}_\lhd g^\sharp)^\flat$. We leave it as a straightforward exercise to verify that $((f)^\Downarrow)^\Uparrow = f$ and $((g)^\Uparrow)^\Downarrow = g$.
We take $\eta := 1^\flat$ (the unit of the adjunction). We have to prove that $(f)^\Downarrow = (\eta \ox f)m$:
\begin{align*}
    (f)^\Downarrow
    &= (a_\lhd f^\sharp)^\flat = (a_\lhd (1\lhd f)\epsilon)^\flat\tag{$f^\sharp = (1\lhd f)\epsilon $}\\
    &= ((1\lhd(1\ox f))a_\lhd\epsilon)^\flat\tag{naturality of $a_\lhd$}\\
    &= (1\ox f)a_\lhd^\flat  \textrm{Hom}(1, \epsilon) = (1\ox f)((1\lhd (\eta\ox\eta)m)\epsilon)^\flat \textrm{Hom}(1, \epsilon)\\
    &= (1\ox f)(\eta\ox\eta)m\epsilon^\flat \textrm{Hom}(1, \epsilon) = (1\ox f)(\eta\ox\eta)m(1^\sharp)^\flat \textrm{Hom}(1, \epsilon)\\
    &= (1\ox f)(\eta\ox\eta)m \textrm{Hom}(1, \epsilon) = (\eta\ox f)(1\ox \eta \textrm{Hom}(1, \epsilon))m\tag{Lemma \ref{swap_m}}\\
    &= (\eta\ox f)m\tag{triangle identity}
\end{align*}
\end{proof}
\section{Categories with copowers are right actegories with hom-objects}
In the previous section, we established that any right actegory equipped with hom-objects admits copowers. In this section, we address the converse direction of Theorem \ref{actegory-w-hom-powers-eqiv}. Before doing so, we first present a lemma on parametrized adjunctions, which will be used extensively throughout the proof.

\begin{lemma}\label{eta-epsilon-dinatural}
    If $F: \bb{Z}\x\bb{X}\to\bb{Y}$ is a functor and if, for each $X\in\bb{X}$, $F(-, X)$ has a right adjoint such that $(\eta, \epsilon):F(-, X) \dashv G(X, -): \bb{Z}\to\bb{Y}$ then $G: \bb{X}^{op}\x\bb{Y}\to\bb{Z}$ is a functor where, for all $f\in\bb{X}$ and $g\in\bb{Y}$, $G(f, g) := (F(1, f)\epsilon g)^\flat$. Furthermore, $\eta$ and $\epsilon$ are dinatural, meaning that $F(f, 1)\epsilon = F(1, G(f, 1))\epsilon$ and $\eta G(f, 1) = \eta G(1, F(f, 1))$
\end{lemma}

\begin{prop}\label{copowers-give-actegories-w-hom-objs}
In any right $\bb{A}$-enriched category $\bb{X}$ with copowers there is, for each $X\in \bb{X}$, an adjunction $X\lhd - \dashv \textrm{Hom}(X, -): \bb{A}\to\bb{X}$ and,  furthermore, $\bb{X}$ is a right $\bb{A}$-actegory with respect to $\lhd$.
\end{prop}
\begin{proof}
First, we have to prove that $\textrm{Hom}(X, -)$ is a functor which is straightforward using $\textrm{Hom}(1, f) = (\utr)^{-1}(1\ox\lc f\rc)m$.
Next we prove that the unit of the copower $\eta: A\to \textrm{Hom}(X, X\lhd A)$ has the universal property.  Using the property of copowers, we can write $\lc f^\sharp \rc = (\utr f)^\Uparrow$.
% $$
% \infer[]{   
%         {
%             I\xrightarrow[]{\lc f^\sharp \rc = (\utr f)^\Uparrow}Hom(X\lhd A, Y)
%         }
%     }
%     {A\ox I \xrightarrow[]{\utr} A \xrightarrow[]{f} \textrm{Hom}(X, y)}
% $$
So we have:
\begin{align*}
    \eta \textrm{Hom}(1, f^\sharp)
    &= (\utr)^{-1}(\eta\ox \lc f^\sharp \rc)m = (\utr)^{-1}(\eta\ox (\utr f)^\Uparrow)m = (\utr)^{-1}((\utr f)^\Uparrow)^\Downarrow = (\utr)^{-1}\utr f = f
\end{align*}
which proves the universal property of $\eta$.
% $$
% \begin{tikzcd}
%     A\arrow[r, "\eta"] \arrow[rd, "f"']
%     & \textrm{Hom}(X, X\lhd A)\\
%     & \textrm{Hom}(X, Y)\arrow[u, "Hom({1, f^\sharp})"']
% \end{tikzcd}
% $$
It remains to prove that for any $f$, $Hom(1, f^\sharp)$ is unique. Suppose there is another map $\textrm{Hom}(1, g): \textrm{Hom}(X, Y)\to \textrm{Hom}(X, X\lhd A)$, such that $\eta \textrm{Hom}(1, g) = f$. We have:
\begin{align*}
    \eta \textrm{Hom}(1, g)
    &= (\utr)^{-1}(\eta \ox \lc g \rc)m = (\utr)^{-1}(\lc g \rc)^\Downarrow = f
\end{align*}
which means $(\lc g \rc)^\Downarrow = \utr f$ and by applying $(-)^\Uparrow$ on both sides we get $\lc g \rc = (\utr f)^\Uparrow = \lc f^\sharp \rc$ which means $\textrm{Hom}(1, g) = \textrm{Hom}(1, f^\sharp)$ and this proves the uniqueness of $\textrm{Hom}(1, f^\sharp)$.
Thus $\eta$ has the universal property. As a result we have an adjunction and $X\lhd -: \bb{A}\to\bb{X}$ is a functor with $\epsilon: X\lhd \textrm{Hom}(X, Y)\to Y$ the counit.  Next we prove that $\bb{X}$ is an $\bb{A}$-actegory.
    $a_\lhd: X\lhd (A\ox B)\to X\lhd A\lhd B$ is defined as:
    \begin{center}
        \scriptsize
        \begin{tikzcd}
            X\lhd (A\ox B) \arrow[rr, "1\lhd(\eta\ox\eta)"]\arrow[rrdd, "a_\lhd"']
            && X\lhd ( \textrm{Hom}(X, X\lhd A)\ox  \textrm{Hom}(X\lhd A, X\lhd A\lhd B)) \arrow[d, "1\lhd m"]\\
            && X\lhd  \textrm{Hom}(X, X\lhd A\lhd B)\arrow[d, "\epsilon"]\\
            && X\lhd A\lhd B
        \end{tikzcd}
    \end{center}
    $a^{-1}_\lhd: X\lhd A\lhd B\to X\lhd (A\ox B)$ can be expressed as a map $\lc a^{-1}_\lhd \rc: I\to \textrm{Hom}(X\lhd A\lhd B, X\lhd (A\ox B))$ in the underlying ordinary category $\lc\bb{X}\rc$, which can be written as $\lc a_\lhd^{-1}\rc = (u^R_\ox (\eta)^\Uparrow)^\Uparrow$ using the property of copowers.
    % $$
    % \infer[]
    % {\infer[]
    %     {
    %     \infer[]
    %         {I\xrightarrow[]{(u^R_\ox (\eta)^\Uparrow)^\Uparrow} \textrm{Hom}(X\lhd A\lhd B, X\lhd (A\ox B))}
    %         {B\ox I\xrightarrow[]{u^R_\ox (\eta)^\Uparrow} \textrm{Hom}(X\lhd A, X\lhd (A\ox B))}
    %     }
    %     {B\xrightarrow[]{(\eta)^\Uparrow} \textrm{Hom}(X\lhd A, X\lhd (A\ox B))}
    % }
    % {A\ox B\xrightarrow[]{\eta} \textrm{Hom}(X, X\lhd (A\ox B))}
    % $$
    \vspace{-0.5cm}
  \begin{align*}
    a_\lhd a^{-1}_\lhd &= ((1\lhd (\eta \ox \eta)m)\epsilon) a^{-1}_\lhd = (1\lhd (\eta\ox\eta)m  \textrm{Hom}(1, a^{-1}_\lhd))\epsilon \tag{naturality of $\epsilon$}\\
    &= (1\lhd (\eta\ox\eta \textrm{Hom}(1, a^{-1}_\lhd))m)\epsilon\tag{lemma \ref{swap_m}} \\
    &= (1\lhd (\eta\ox((\utr)^{-1}(\eta\ox\lc a^{-1}_\lhd\rc)m))m)\epsilon \tag{corollary \ref{updown-all}}\\ 
    &= (1\lhd (\eta\ox((\utr)^{-1}(\lc a^{-1}_\lhd\rc)^\Downarrow)m)\epsilon \tag{definition \ref{copowers-def}}\\
    &= (1\lhd (\eta\ox((\utr)^{-1}((\utr(\eta)^\Uparrow)^\Uparrow)^\Downarrow)m)\epsilon \tag{def of $\lc a^{-1}_\lhd \rc$}\\
    &= (1\lhd (\eta\ox((\utr)^{-1}\utr(\eta)^\Uparrow)m)\epsilon = (1\lhd (\eta\ox (\eta)^\Uparrow) m)\epsilon\\
    &= (1\lhd ((\eta)^\Uparrow)^\Downarrow)\epsilon \tag{definition \ref{copowers-def}}\\
    &= (1\lhd \eta)\epsilon = 1 \tag{triangle identity}
  \end{align*}
\vspace{-1cm}
  \begin{align*}
    a^{-1}_\lhd a_\lhd &= (1\lhd ((\utr)^{-1} (\eta \ox \lc a^{-1}_\lhd\rc)m ))\epsilon a_\lhd \\
    &= (1\lhd ((\utr)^{-1} (\eta \ox \lc a^{-1}_\lhd\rc)m  \textrm{Hom}(1, a_\lhd))\epsilon \tag{naturality of $\epsilon$}\\
    &= (1\lhd ((\utr)^{-1} ((\eta \ox \lc a^{-1}_\lhd\rc \textrm{Hom}(1, a^{-1}_\lhd))m)\epsilon\tag{lemma \ref{swap_m}}\\
    &= (1\lhd ((\utr)^{-1} ((\eta \ox (\utr(\eta)^\Uparrow)^\Uparrow \textrm{Hom}(1, a_\lhd))m)\epsilon \tag{def of $\lc a^{-1}_\lhd \rc$}\\ 
    &= (1\lhd ((\utr)^{-1} ((\eta \ox (\utr(\eta)^\Uparrow \textrm{Hom}(1, a_\lhd))^\Uparrow)m)\epsilon \tag{corollary \ref{updown-all}}\\
    &= (1\lhd ((\utr)^{-1} ((\eta \ox (\utr(\eta \textrm{Hom}(1, a_\lhd))^\Uparrow)^\Uparrow)m)\epsilon \tag{corollary \ref{updown-all}}\\
    &= (1\lhd ((\utr)^{-1} ((\eta \ox (\utr(\eta \textrm{Hom}(1, (1\lhd (\eta)^\Downarrow)\epsilon)))^\Uparrow)^\Uparrow)m)\epsilon \tag{def of $a_\lhd$}\\
    &= (1\lhd ((\utr)^{-1} ((\eta \ox (\utr((\eta)^\Downarrow \eta \textrm{Hom}(1,\epsilon)))^\Uparrow)^\Uparrow)m)\epsilon\tag{naturality of $\eta$}\\
    &= (1\lhd ((\utr)^{-1} ((\eta \ox (\utr((\eta)^\Downarrow)^\Uparrow)^\Uparrow)m)\epsilon\tag{triangle identity}\\
    &= (1\lhd ((\utr)^{-1} ((\eta \ox (\utr\eta)^\Uparrow)m)\epsilon = (1\lhd ((\utr)^{-1}((\utr\eta)^\Uparrow)^\Downarrow)\epsilon  \tag{definition \ref{copowers-def}}\\
    &= (1\lhd ((\utr)^{-1}\utr\eta)\epsilon = (1\lhd\eta)\epsilon = 1  \tag{triangle identity}\\
  \end{align*}
  So we can conclude that $a_\lhd$ is an isomorphism.
  In order to have an actegory, the coherence diagrams for $a_\lhd$ (see Definition \ref{actegory-def}) must hold. Based on the definition of $a_\lhd$, on the left-hand-side of the middle diagram of \ref{actegory-def} we have:
  \begin{align*}
      (1\lhd a_\ox) a_\lhd a_\lhd
      &= (1\lhd a_\ox)(1\lhd (\eta\ox\eta)m)\epsilon a_\lhd = (1\lhd a_\ox(\eta\ox\eta)m)\epsilon a_\lhd\\
      &= (1\lhd a_\ox(\eta\ox\eta)m \textrm{Hom}(1, a_\lhd))\epsilon\tag{naturality of $\epsilon$}\\
      &= (1\lhd a_\ox(\eta\ox\eta \textrm{Hom}(1, a_\lhd))m)\epsilon\tag{Lemma \ref{swap_m}}\\
      &= (1\lhd a_\ox(\eta\ox\eta \textrm{Hom}(1, (1\lhd(\eta\ox\eta)m)\epsilon))m)\epsilon\tag{def of $a_\lhd$}\\
      &= (1\lhd a_\ox(\eta\ox(\eta\ox\eta)m\eta \textrm{Hom}(1,\epsilon))m)\epsilon \tag{naturality of $\eta$}\\
      &= (1\lhd a_\ox(\eta\ox(\eta\ox\eta)m)m)\epsilon \tag{triangle identity}\\
  \end{align*}
  And on the right-hand-side we have:
  \begin{align*}
      a_\lhd (a_\lhd \lhd 1)
      &= (1\lhd (\eta\ox\eta)m)\epsilon (a_\lhd \lhd 1) \tag{definition of $a_\lhd$}\\
      &= (1\lhd (\eta\ox\eta)m \textrm{Hom}(1, a_\lhd \lhd 1))\epsilon \tag{naturality of $\epsilon$}\\
      &= (1\lhd (\eta\ox\eta \textrm{Hom}(1, a_\lhd \lhd 1))m)\epsilon \tag{Lemma \ref{swap_m}}\\
      &= (1\lhd (\eta\ox\eta \textrm{Hom}(a_\lhd, 1))m)\epsilon \tag{lemma \ref{eta-epsilon-dinatural}}\\
      &= (1\lhd (\eta\ox((\utr)^{-1}(\lc a_\lhd\rc\ox\eta)m))m)\epsilon\\
      &= (1\lhd ((1\ox(\utr)^{-1})(\eta\ox(\lc a_\lhd\rc\ox\eta))m)m)\epsilon\\
      &= (1\lhd ((\utr)^{-1}(\eta\ox\lc a_\lhd\rc)m\ox\eta))m)\epsilon \tag{associativity of $m$}\\
      &= (1\lhd ((\eta \textrm{Hom}(1, a_\lhd)\ox\eta))m)\epsilon\\
      &= (1\lhd ((\eta \textrm{Hom}(1, (1\lhd(\eta\ox\eta)m)\epsilon)\ox\eta))m)\epsilon \tag{def of $a_\lhd$}\\
      &= (1\lhd (((\eta\ox\eta)m\eta \textrm{Hom}(1, \epsilon)\ox\eta))m)\epsilon \tag{naturality of $\eta$}\\
      &= (1\lhd (((\eta\ox\eta)m\ox\eta))m)\epsilon \tag{triangle identity}\\      
      &= (1\lhd (((\eta\ox(\eta\ox\eta)m)m)\epsilon \tag{associativity of $m$}\\
  \end{align*}
  So the right-hand-side and the left-hand-side of the middle diagram are equal and thus the coherence for $a_\lhd$ holds. In addition, $a_\lhd$ needs to be natural in every argument. So the following diagrams must commute:
  \vspace*{-0.25cm}
  \begin{center}
    \scriptsize
        \begin{tikzcd}
          X\lhd (A\ox B) \arrow[rr, "a_\lhd"{name=A}] \arrow[dd, "1\lhd (1\ox h)"']
          && (X\lhd A) \lhd B \arrow[dd, "1\lhd h"]\\\\
          X\lhd (A\ox B')\arrow[rr, "a_\lhd"'{name=B}]\arrow[, from=A, to=B, eq=eq:h]
          &&(X\lhd A)\lhd B'
        \end{tikzcd}
        \begin{tikzcd}
          X\lhd (A\ox B) \arrow[rr, "a_\lhd"{name=A}] \arrow[dd, "1\lhd (g\ox1)"']
          && (X\lhd A) \lhd B \arrow[dd, "(1\lhd g) \lhd 1"]\\\\
          X\lhd (A'\ox B)\arrow[rr, "a_\lhd"'{name=B}]\arrow[,from=A, to=B, eq=eq:g]
          &&(X\lhd A')\lhd B
        \end{tikzcd}
        \begin{tikzcd}
          X\lhd (A\ox B) \arrow[rr, "a_\lhd"{name=A}] \arrow[dd, "f\lhd 1"']
          && (X\lhd A) \lhd B \arrow[dd, "(f\lhd 1) \lhd 1"]\\\\
          Y\lhd (A\ox B)\arrow[rr, "a_\lhd"'{name=B}]\arrow[,from=A, to=B, eq=eq:f]
          &&(Y\lhd A)\lhd B
        \end{tikzcd}
  \end{center}
  % \vspace{-0.25cm}
   
    % \vspace{-0.85cm}
    \begin{align*}
        {\rm For}~(\ref{eq:f}):~~a_\lhd ((f\lhd 1)\lhd 1) 
        &= (1\lhd(\eta\ox\eta)m)\epsilon ((f\lhd 1)\lhd 1)\\
        &= (1\lhd(\eta\ox\eta)m  \textrm{Hom}(1, (f\lhd 1)\lhd 1))\epsilon\tag{naturality of $\epsilon$}\\
        &= (1\lhd(\eta\ox(\eta \textrm{Hom}(1, (f\lhd 1)\lhd 1)))m)\epsilon\tag{lemma \ref{swap_m}}\\
        &= (1\lhd(\eta\ox(\eta \textrm{Hom}(f\lhd 1, 1)))m)\epsilon\tag{lemma \ref{eta-epsilon-dinatural}}\\
        &= (1\lhd(\eta\ox((\utl)^{-1}(\lc f\lhd 1\rc\ox\eta)m))m)\epsilon\\
        &= (1\lhd(1\ox (\utl)^{-1})(\eta\ox((\lc f\lhd 1\rc\ox\eta)m))m)\epsilon\\
        &= (1\lhd(1\ox (\utl)^{-1})a_\ox((\eta\ox\lc f\lhd 1\rc)m \ox\eta)m)\epsilon \tag{associativity of $m$}\\
        &= (1\lhd((\utr)^{-1}(\eta\ox\lc f\lhd 1\rc)m \ox\eta)m)\epsilon\\
        &= (1\lhd((\eta \textrm{Hom}(1, f\lhd 1))\ox\eta)m)\epsilon\\
        &= (1\lhd((\eta \textrm{Hom}(f, 1))\ox\eta)m)\epsilon\tag{lemma \ref{eta-epsilon-dinatural}}\\
        &= (1\lhd(((\utr)^{-1}(\lc f\rc\ox\eta)m)\ox\eta)m)\epsilon\\
        &= (1\lhd(\utr)^{-1}(\lc f\rc\ox(\eta\ox\eta)m)m)\epsilon\tag{associativity of $m$}\\
        &= (1\lhd(\utr)^{-1}(\lc f\rc\ox(\eta\ox\eta)m\eta \textrm{Hom}(1, \epsilon))m)\epsilon\tag{triangle identity}\\
        &= (1\lhd(\utr)^{-1}(\lc f\rc\ox\eta \textrm{Hom}(1, 1\lhd ((\eta\ox\eta)m)\epsilon))m)\epsilon\tag{naturality of $\eta$}\\
        &= (1\lhd(\utr)^{-1}(\lc f\rc\ox\eta \textrm{Hom}(1, a_\lhd))m)\epsilon\tag{def of $a_\lhd$}\\
        &= (1\lhd(\utr)^{-1}(\lc f\rc\ox\eta)m \textrm{Hom}(1, a_\lhd))\epsilon\tag{lemma \ref{swap_m}}\\
        &= (1\lhd(\utr)^{-1}(\lc f\rc\ox\eta)m)\epsilon a_\lhd\tag{naturality of $\epsilon$}\\
        &= (1\lhd(\eta \textrm{Hom}(f, 1)))\epsilon a_\lhd= (f\lhd 1) a_\lhd\\
    \end{align*}
    
     \begin{align*}
        {\rm For}~(\ref{eq:h}):~~a_\lhd (1\lhd h) 
        &= (1\lhd (\eta\ox\eta)m)\epsilon (1\lhd h) = (1\lhd (\eta\ox\eta)m \textrm{Hom}(1, 1\lhd h))\epsilon\tag{naturality of $\epsilon$}\\
        &= (1\lhd (\eta\ox(\eta \textrm{Hom}(1, 1\lhd h)))m)\epsilon\\
        &= (1\lhd(\eta\ox(h\eta))m)\epsilon \tag{naturality of $\eta$} = (1\lhd (1\ox h)(\eta\ox\eta)m)\epsilon\\
        &= (1\lhd (1\ox h))(1\lhd (\eta\ox\eta)m)\epsilon = (1\lhd (1\ox h))a_\lhd \tag{def of $a_\lhd$}\\
    \end{align*}
    % \vspace{-0.85cm}
    \begin{align*}
        {\rm For}~(\ref{eq:g}):~~a_\lhd ((1\lhd g) \lhd 1)
        &= (1\lhd (\eta\ox\eta)m)\epsilon ((1\lhd g)\lhd 1)\\
        &= (1\lhd (\eta\ox\eta)m \textrm{Hom}(1, ((1\lhd g)\lhd 1)))\epsilon\tag{naturality of $\epsilon$}\\
        &= (1\lhd (\eta\ox(\eta \textrm{Hom}(1, ((1\lhd g)\lhd 1))))m)\epsilon \tag{lemma \ref{swap_m}}\\
        &= (1\lhd (\eta\ox(\eta \textrm{Hom}((1\lhd g), 1)))m)\epsilon \tag{lemma \ref{eta-epsilon-dinatural}}\\
        &= (1\lhd (\eta\ox((\utr)^{-1}(\lc 1\lhd g\rc\ox\eta)m))m)\epsilon\\
        &= (1\lhd (((\utr)^{-1}(\eta\ox\lc 1\lhd g\rc)m)\ox\eta))m)\epsilon\tag{associativity of $m$}\\
        &= (1\lhd (\eta \textrm{Hom}(1, 1\lhd g)\ox\eta)m)\epsilon = (1\lhd (g\eta\ox\eta)m)\epsilon\tag{naturality of $\eta$}\\
        &= (1\lhd (g\ox1)(\eta\ox\eta)m)\epsilon = (1\lhd (g\ox 1))a_\lhd\\
    \end{align*}
    % \vspace{-0.1cm}
    So $a_\lhd$ is natural. The other map which is required for forming an actegory is $u_\lhd: X\lhd I\to I$ which is defined by $u_\lhd := (1\lhd id)\epsilon$.
    % \begin{center}
    %     \begin{tikzcd}
    %         X\lhd I \arrow[r, "1\lhd id"]\arrow[rd, "u_\lhd"']
    %         & X\lhd  \textrm{Hom}(X, X)\arrow[d, "\epsilon"]\\
    %         & X
    %     \end{tikzcd}     
    % \end{center}
    $u_\lhd$ can also be expressed as a map $\lc u_\lhd \rc: I \to \textrm{Hom}(X\lhd I, X)$ in the category $\lc\bb{X}\rc$, which can be defined using the property of copowers as $u_\lhd := (\utl id)^\Uparrow$.
    % $$
    % \infer[]
    % {I\ox I \xrightarrow[]{\utl}I\xrightarrow[]{id} \textrm{Hom}(X, X)}
    % {I\xrightarrow[]{u_\lhd = (\utl id)^\Uparrow} \textrm{Hom}(X\lhd I, X)}
    % $$
    The name of the inverse of the $u_\lhd$ is $u^{-1}_\lhd:= \eta : X\to X\lhd I$.   
    We must show that $u^{-1}_\lhd$ is inverse to $u_\lhd$ that is:

    \begin{align*}
        u_\lhd u^{-1}_\lhd
        &= (\utr)^{-1}(\lc u_\lhd \rc \ox \lc u^{-1}_\lhd \rc)m = (\utr)^{-1}((\utl id)^\Uparrow \ox \eta)m\\
        &= (\utr)^{-1}(((\utl id)\ox \eta)m)^\Uparrow = ((1\ox (\utr)^{-1})(\utl\ox 1)(id\ox\eta)m)^\Uparrow\\
        &= (\utl(\utl)^{-1}(id\ox\eta)m)^\Uparrow = (\utl\eta)^\Uparrow = id\\
        u^{-1}_\lhd u_\lhd
        &= (\utl)^{-1}(\lc u^{-1}_\lhd \rc \ox \lc u_\lhd \rc )m = (\utl)^{-1}(\eta \ox (\utl id)^\Uparrow)m = (\utl)^{-1}((\utl id)^\Uparrow)^\Downarrow = (\utl)^{-1}\utl id = id
    \end{align*}
    Which shows that $u_\lhd$ is an isomorphism. Next, we must show that $u_\lhd$ satisfies the left and the right coherence diagrams of Definition \ref{actegory-def}:
    \begin{align*}
        a_\lhd u_\lhd
        &= (1\lhd(\eta\ox\eta)m)\epsilon u_\lhd = (1\lhd(\eta\ox\eta)m \textrm{Hom}(1, u_\lhd))\epsilon\tag{naturality of $\epsilon$}\\
        &= (1\lhd(\eta\ox(\eta \textrm{Hom}(1, u_\lhd)))m)\epsilon\tag{lemma \ref{swap_m}}\\
        &= (1\lhd(\eta\ox(\eta \textrm{Hom}(1, (1\lhd id)\epsilon)))m)\epsilon\tag{def of $u_\lhd$}\\
        &= (1\lhd(\eta\ox(id\eta \textrm{Hom}(1, \epsilon)))m)\epsilon\tag{naturality of $\eta$}\\
        &= (1\lhd(\eta\ox id)m)\epsilon\tag{triangle identity}\\
        &= (1\lhd \utr \eta)\epsilon = (1\lhd \utr)(1\lhd \eta)\epsilon\\
        &= 1\lhd\utr\tag{triangle identity}\\
    \end{align*}
    \begin{align*}
        a_\lhd (u_\lhd\lhd 1)
        &= (1\lhd(\eta\ox\eta)m)\epsilon (u_\lhd\lhd 1)\\
        &= (1\lhd(\eta\ox\eta)m \textrm{Hom}(1, u_\lhd\lhd 1))\epsilon\tag{naturality of $\epsilon$}\\
        &= (1\lhd(\eta\ox(\eta \textrm{Hom}(1, u_\lhd\lhd 1)))m)\epsilon\tag{lemma \ref{swap_m}}\\
        &= (1\lhd(\eta\ox(\eta \textrm{Hom}(u_\lhd, 1)))m)\epsilon\tag{lemma \ref{eta-epsilon-dinatural}}\\
        &= (1\lhd(\eta\ox ((\utl)^{-1}(\lc u_\lhd \rc \ox \eta)m))m)\epsilon\\
        &= (1\lhd(1\ox (\utl)^{-1})(\eta\ox (\lc u_\lhd \rc \ox \eta)m)m)\epsilon\\
        &= (1\lhd(1\ox (\utl)^{-1})a_\ox((\eta\ox\lc u_\lhd\rc)m \ox \eta)m)\epsilon\tag{associativity of $m$}\\
        &= (1\lhd((\utr)^{-1}\ox 1)((\eta\ox\lc u_\lhd\rc)m \ox \eta)m)\epsilon\\
        &= (1\lhd((\utr)^{-1}(\eta\ox\lc u_\lhd\rc)m \ox \eta)m)\epsilon\\
        &= (1\lhd(\eta \textrm{Hom}(1, u_\lhd)\ox\eta)m)\epsilon\\
        &= (1\lhd(\eta \textrm{Hom}(1, (1\lhd id)\epsilon)\ox\eta)m)\epsilon\tag{def of $u_\lhd$}\\
        &= (1\lhd(id\eta \textrm{Hom}(1, \epsilon)\ox\eta)m)\epsilon\tag{naturality of $\eta$}\\
        &= (1\lhd(id\ox\eta)m)\epsilon\tag{triangle identity}\\
        &= (1\lhd\utl\eta)\epsilon = (1\lhd\utl)(1\lhd\eta)\epsilon\\ &= 1\lhd\utl\tag{triangle identity}
    \end{align*}
    \\
    Finally, $u_\lhd$ has to be natural, so $u\lhd f = (f\lhd 1)u\lhd$
    % \begin{center}        
    %     \begin{tikzcd}
    %         X\lhd I\arrow[r, "u_\lhd"]\arrow[d, "f\lhd 1"']
    %         & X\arrow[d, "f"]\\
    %         Y\lhd I\arrow[r, "u_\lhd"']
    %         & Y
    %     \end{tikzcd}
    % \end{center}
    \begin{align*}
        (f\lhd 1)u_\lhd
        &= (1\lhd (\eta \textrm{Hom}(f, 1)))\epsilon u_\lhd\\
        &= (1\lhd ((\utr)^{-1}(\lc f\rc \ox \eta)m))\epsilon u_\lhd\\
        &= (1\lhd (((\utr)^{-1}(\lc f\rc \ox \eta)m) \textrm{Hom}(1, u_\lhd)))\epsilon\tag{naturality of $\epsilon$}\\
        &= (1\lhd (((\utr)^{-1}(\lc f\rc \ox (\eta \textrm{Hom}(1, u_\lhd)))m)))\epsilon\tag{lemma \ref{swap_m}}\\
        &= (1\lhd (((\utr)^{-1}(\lc f\rc \ox (\eta \textrm{Hom}(1, (1\lhd id)\epsilon )))m)))\epsilon\tag{def of $u_\lhd$}\\
        &= (1\lhd (((\utr)^{-1}(\lc f\rc \ox (id\eta \textrm{Hom}(1, \epsilon)))m)))\epsilon\tag{naturality of $\eta$}\\
        &= (1\lhd (((\utr)^{-1}(\lc f\rc \ox id)m)))\epsilon\tag{triangle identity}\\
        &= (1\lhd (((\utr)^{-1}(id\ox\lc f\rc)m)))\epsilon = (1\lhd id \textrm{Hom}(1, f) )\epsilon\\
        &= (1\lhd id)\epsilon f \tag{naturality of $\epsilon$} = u_\lhd f
    \end{align*}    
\end{proof}
Propositions \ref{actegories-w-hom-objs-give-copowers} and \ref{copowers-give-actegories-w-hom-objs} allow us to state the main theorem of the paper:

\begin{theorem}\label{actegory-w-hom-copowers-equivalence}
To give a right actegory with (right) hom-objects is to give a right enriched category with copowers.
\end{theorem}

\section{Powers}
\begin{definition}\label{powers-def}
    Dual to copowers, a left $\bb{A}$-enriched category $\bb{X}$ has \textbf{powers} if, for each $A\in\bb{A}$ and $Y\in\bb{X}$, there exists an object $A\brhd Y\in \bb{X}$ and a morphism $\epsilon: A\to  \textrm{Hom}(A\brhd Y, Y)$ in $\bb{A}$ (called the counit of the power) such that for any $B\in\bb{A}$ and $X\in\bb{X}$ there is a bijective correspondence:
    $$
    \infer=[]{B\otimes A \xrightarrow[]{(f\ox\epsilon)m}\textrm{Hom}(X, Y)}{B \xrightarrow[]{f}\textrm{Hom}(X, A\brhd Y)}
    $$
\end{definition}
\begin{theorem}[Dual of theorem \ref{actegory-w-hom-copowers-equivalence}]\label{actegory-w-hom-powers-eqiv}
    To give a left actegory with hom-objects is to give a left enriched category with powers.
\end{theorem}
We now show that in a right actegory with hom-objects, if the action functor is a parametrized left adjoint (with parametrized right adjoint $A\brhd -$), the right adjoint action yields powers just as the action yields copowers. More importantly, in this case, both powers and copowers have isomorphic hom-object, demonstrating that they induce the same enrichment up to equivalence.
\begin{theorem}\label{wb-adjunction-give-op-hom-objects}
A right $\bb{A}$-actegory $\bb{X}$ with hom-objects is a left $(\bb{A}^{rev})^{op}$-actegory with hom-objects if and only if $\forall A \in \bb{A}$, the action functor has a right adjoint as $-\lhd A \dashv A\brhd - :\bb{X}\to\bb{X}$.
\end{theorem}
\begin{proof}
     We start with proving a right $\bb{A}$-actegory $\bb{X}$ with right hom-objects is a left $(\bb{A}^{rev})^{op}$-actegory with left hom-objects provided that the action functor has a right adjoint $(\eta', \epsilon'): -\lhd A \dashv A\brhd - :\bb{X}\to\bb{X}$. Since $\bb{X}$ is a right $\bb{A}$-actegory with hom-objects there is an adjunction $(\eta, \epsilon): X\lhd - \dashv \textrm{Hom}(X, -) : \bb{A}\to \bb{X}$. Using the two adjunctions we obtain:
    % $$
    % \infer=[]{
    % X \lhd A \xrightarrow[]{(1\lhd f)\epsilon} Y
    % }{
    % A\xrightarrow[]{f} \textrm{Hom} (X, Y)
    % }
    % $$
    % We know that there is another adjunction $(\eta', \epsilon'): - \lhd A\dashv A \brhd - : \lc\bb{X}\rc\to\lc\bb{X}\rc$. So we have:
    % $$
    % \infer=[]{
    %     \infer=[]{X\xrightarrow[]{\eta'(1\brhd ((1\lhd f)\epsilon))} A\brhd Y}{X \lhd A \xrightarrow[]{(1\lhd f)\epsilon} Y}
    % }{
    % A\xrightarrow[]{f} \textrm{Hom}(X, Y)
    % }
    % $$
    % Since $A\to \textrm{Hom}(X, Y)$ is a map in $\bb{A}$, $\textrm{Hom}(X, Y)\to A$ is a map in $(\bb{A}^{rev})^{op}$ so we can write:
    $$
    \infer=[]{
        \infer=[]{X\xrightarrow[]{\eta'(1\brhd ((1\lhd f)\epsilon))} A\brhd Y}{X \lhd A \xrightarrow[]{(1\lhd f)\epsilon} Y}
    }{
    \infer=[]{
        A\xrightarrow[]{f}\textrm{Hom}(X, Y) \quad in~\bb{A}}{\textrm{Hom}(X, Y)\xrightarrow[]{f^{op}} A & in~(\bb{A}^{rev})^{op}}
    }
    $$
    So we can define the $(-)^\flat$ operator as $f^\flat = \eta'(1\brhd ((1\lhd f)\epsilon))$.
    % $$
    % \infer=[]{
    %     X\xrightarrow[]{f^\flat = \eta'(1\brhd ((1\lhd f)\epsilon))} A\brhd Y
    % }{
    % \textrm{Hom}(X, Y)\xrightarrow[]{f^{op}} A
    % }
    % $$
    It suffices to prove that $(\textrm{Hom}(k, 1) f^{op} h)^\flat$ $= k f^\flat (h\brhd 1) $ which means $(h f \textrm{Hom}(k, 1))^\flat = k f^\flat (h\brhd 1) $. We have:
    \begin{align*}
        (h f \textrm{Hom}(k, 1))^\flat
        &= \eta' (1\brhd ((1\lhd h f \textrm{Hom}(k, 1))\epsilon))\\
        &= \eta' (h \brhd (1\lhd f \textrm{Hom}(k, 1))\epsilon)\tag{dinaturality of $\eta'$(lemma \ref{eta-epsilon-dinatural})}\\
        &= \eta' (h \brhd ((k \lhd f)\epsilon)) = k (\eta' (h\brhd ((1\lhd f)\epsilon)))\\
        &= k (\eta' (1\brhd ((1\lhd f)\epsilon)))(h\brhd 1) = k f^\flat (h\brhd 1)\\
    \end{align*}
    And consequently $\textrm{Hom}(X, -)\dashv -\brhd X: \bb{X}\to(\bb{A}^{rev})^{op}$ which means that $(\bb{A}^{rev})^{op}$-actegory has hom-objects. Now we have to prove the other direction of the theorem. Using the two adjunctions $(\eta, \epsilon): X\lhd - \dashv \textrm{Hom}(X, -) : \bb{A}\to \bb{X}$ and $(\eta'', \epsilon''):\textrm{Hom}(X, -)\dashv -\brhd X: \bb{X}\to(\bb{A}^{rev})^{op}$ we have: 
    $$
    \infer=[]{
    X\xrightarrow[]{\eta'(\textrm{Hom}(f, 1)\eta\brhd 1)} A\brhd Y
    }{
    \infer=[]{
        \textrm{Hom}(X, Y)\xrightarrow[]{\textrm{Hom}(f, 1)\eta}A \quad in~(\bb{A}^{rev})^{op}}{
        \infer=[]{A\xrightarrow[]{\eta\textrm{Hom}(1, f)}\textrm{Hom}(X, Y) \quad in~\bb{A}}{
            X\lhd A\xrightarrow[]{f}Y
        }
        }
    }
    $$
    So we can define the $(-)^\flat$ operator as $f^\flat = \eta'(\textrm{Hom}(f, 1)\eta\brhd 1)$. It suffices to show that $((h\lhd 1)fk)^\flat = hf^\flat(1\brhd k)$. We have:
    \begin{align*}
        hf^\flat(1\brhd k)
        & = h\eta'(\textrm{Hom}(f, 1)\eta\brhd 1)(1\brhd k) = h\eta'(1\brhd k)(\textrm{Hom}(f, 1)\eta\brhd 1)\\
        & = h\eta'(\textrm{Hom}(kf, 1)\eta\brhd 1)\tag{dinaturality of $\eta'$ (lemma \ref{eta-epsilon-dinatural})}\\
        &= \eta'(\textrm{Hom}(kf, 1)\textrm{Hom}(1, h)\eta\brhd 1)\tag{naturality of $\eta'$}\\
        &= \eta'(\textrm{Hom}(kf(h\lhd 1), 1)\eta\brhd 1)\tag{dinaturality of $\eta$ (lemma \ref{eta-epsilon-dinatural})}\\
        &= ((h\lhd 1)fk)^\flat
    \end{align*}
    And this proves that there is an adjunction $-\lhd A \dashv A\brhd - : \bb{X}\to\bb{X}$.
\end{proof}
Using Theorems \ref{wb-adjunction-give-op-hom-objects} and \ref{actegory-w-hom-powers-eqiv}, we can conclude that \emph{a category is both powered and copowered, if and only if the power and the copower form an adjoint pair.}

%\bibliographystyle{eptcs} 
%\bibliography{references}

\nocite{*}
\begin{thebibliography}{10}
\providecommand{\bibitemdeclare}[2]{}
\providecommand{\surnamestart}{}
\providecommand{\surnameend}{}
\providecommand{\urlprefix}{Available at }
\providecommand{\url}[1]{\texttt{#1}}
\providecommand{\href}[2]{\texttt{#2}}
\providecommand{\urlalt}[2]{\href{#1}{#2}}
\providecommand{\doi}[1]{doi:\urlalt{https://doi.org/#1}{#1}}
\providecommand{\eprint}[1]{arXiv:\urlalt{https://arxiv.org/abs/#1}{#1}}
\providecommand{\bibinfo}[2]{#2}

\bibitemdeclare{book}{borceux1994handbook}
\bibitem{borceux1994handbook}
\bibinfo{author}{F.~\surnamestart Borceux\surnameend} (\bibinfo{year}{1994}):
  \emph{\bibinfo{title}{Handbook of categorical algebra: Categories and
  structures}}.
\newblock \bibinfo{volume}{2}, \bibinfo{publisher}{Cambridge Univ. Press}.

\bibitemdeclare{article}{actegoriesAmthmatician}
\bibitem{actegoriesAmthmatician}
\bibinfo{author}{M.~\surnamestart Capucci\surnameend} \&
  \bibinfo{author}{B.~\surnamestart Gavranovi\'c\surnameend}
  (\bibinfo{year}{2022}): \emph{\bibinfo{title}{Actegories for the Working
  Amthematician}}.
\newblock \doi{10.48550/arXiv.2203.16351}.

\bibitemdeclare{article}{logicOfMessagePassing}
\bibitem{logicOfMessagePassing}
\bibinfo{author}{J.R.B. \surnamestart Cockett\surnameend} \&
  \bibinfo{author}{C.~\surnamestart Pastro\surnameend} (\bibinfo{year}{2009}):
  \emph{\bibinfo{title}{The Logic of Message Passing}}.
\newblock {\slshape \bibinfo{journal}{Science of Computer Programming}}
  \bibinfo{volume}{74}, \doi{10.1016/j.scico 2007.11.005}.

\bibitemdeclare{article}{COCKETT1997133}
\bibitem{COCKETT1997133}
\bibinfo{author}{J.R.B. \surnamestart Cockett\surnameend} \&
  \bibinfo{author}{R.A.G. \surnamestart Seely\surnameend}
  (\bibinfo{year}{1997}): \emph{\bibinfo{title}{Weakly distributive
  categories}}.
\newblock {\slshape \bibinfo{journal}{Journal of Pure and Applied Algebra}}
  \bibinfo{volume}{114}(\bibinfo{number}{2}), pp. \bibinfo{pages}{133--173},
  \doi{10.1016/0022-4049(95)00160-3}.

\bibitemdeclare{article}{GORDON1997167}
\bibitem{GORDON1997167}
\bibinfo{author}{R.~\surnamestart Gordon\surnameend} \& \bibinfo{author}{J.A.
  \surnamestart Power\surnameend} (\bibinfo{year}{1997}):
  \emph{\bibinfo{title}{Enrichment through variation}}.
\newblock {\slshape \bibinfo{journal}{Journal of Pure and Applied Algebra}}
  \bibinfo{volume}{120}(\bibinfo{number}{2}), pp. \bibinfo{pages}{167--185},
  \doi{10.1016/S0022-4049(97)00070-4}.

\bibitemdeclare{article}{aNoteActionsOfAMonoidalCategory}
\bibitem{aNoteActionsOfAMonoidalCategory}
\bibinfo{author}{G.~\surnamestart Janelidze\surnameend} \&
  \bibinfo{author}{G.~M. \surnamestart Kelly\surnameend}
  (\bibinfo{year}{2001}): \emph{\bibinfo{title}{A Note on Actions of a Monoidal
  Category}}.
\newblock {\slshape \bibinfo{journal}{Theory and Applications of Categories}}
  \bibinfo{volume}{9}(\bibinfo{number}{4}).

\bibitemdeclare{book}{kelly1982basic}
\bibitem{kelly1982basic}
\bibinfo{author}{G.M. \surnamestart Kelly\surnameend} (\bibinfo{year}{1982}):
  \emph{\bibinfo{title}{Basic concepts of enriched category theory}}.
\newblock \bibinfo{volume}{64}, \bibinfo{publisher}{CUP Archive}.

\bibitemdeclare{mastersthesis}{ImplementationOfMPL}
\bibitem{ImplementationOfMPL}
\bibinfo{author}{P.~\surnamestart Kumar\surnameend} (\bibinfo{year}{2018}):
  \emph{\bibinfo{title}{Implementation of Message Passing Language}}.
\newblock Master's thesis, \bibinfo{school}{University of Calgary, Calgary,
  Canada}, \doi{10.11575/PRISM/5476}.

\bibitemdeclare{article}{lucyshynv}
\bibitem{lucyshynv}
\bibinfo{author}{R.~\surnamestart Lucyshyn-Wright\surnameend}
  (\bibinfo{year}{2024}): \emph{\bibinfo{title}{V-graded categories and
  V-W-bigraded categories: Functor categories and bifunctors over non-symmetric
  bases}}.
\newblock \urlprefix\url{arXiv:2502.18557}.

\bibitemdeclare{article}{pare2012mealy}
\bibitem{pare2012mealy}
\bibinfo{author}{R.~\surnamestart Par{\'e}\surnameend} (\bibinfo{year}{2012}):
  \emph{\bibinfo{title}{Mealy morphisms of enriched categories}}.
\newblock {\slshape \bibinfo{journal}{Applied Categorical Structures}}
  \bibinfo{volume}{20}(\bibinfo{number}{3}), pp. \bibinfo{pages}{251--273},
  \doi{10.1007/s10485-010-9238-8}.

\bibitemdeclare{misc}{jaredpon}
\bibitem{jaredpon}
\bibinfo{author}{J.~\surnamestart Pon\surnameend} (\bibinfo{year}{2021}):
  \emph{\bibinfo{title}{Redesigning the Abstract Machine for CaMPL}}.
\newblock \urlprefix\url{https://campl-ucalgary.github.io/}.

\end{thebibliography}
\end{document}